\documentclass[
    a4paper,
    USenglish,
    cleveref,
    autoref,
    thm-restate,
]{lipics-v2021}

\nolinenumbers
\hideLIPIcs

\usepackage{algorithm}
\usepackage{algpseudocode}
\usepackage{cases}
\usepackage{mathtools}
\usepackage{amsmath}
\usepackage{complexity}

\newcommand{\eps}{\varepsilon}
\newcommand{\sqeps}{\sqrt\varepsilon}
\newcommand{\reveps}{\frac{1}{\eps}}
\newcommand{\floor}[1]{\left\lfloor #1 \right\rfloor}
\newcommand{\ceil}[1]{\left\lceil #1 \right\rceil}
\newcommand{\prb}[1]{\textnormal{\scshape #1}}

\usepackage{subcaption}
\usepackage{tabularx}
\usepackage{booktabs} 

\newcommand{\defproblem}[3]{
  \vspace{2mm}
  \vspace{1mm}
\noindent\fbox{
  \begin{minipage}{0.95\textwidth}
  #1 \\
  {\bf{Input:}} #2  \\
  {\bf{Task:}} #3
  \end{minipage}
  }
  \vspace{2mm}
}

\usepackage[textsize=tiny]{todonotes}

\title{Approximate Min-Sum Subset Convolution} 

\author{Mihail Stoian}{University of Technology Nuremberg, Germany\and \url{https://stoianmihail.github.io} }{mihail.stoian@utn.de}{https://orcid.org/0000-0002-8843-3374}{}

\authorrunning{M. Stoian} 

\Copyright{Mihail Stoian} 

\ccsdesc[500]{Theory of computation~Dynamic programming}
\ccsdesc[500]{Theory of computation~Approximation algorithms analysis}

\keywords{subset convolution,
    min-plus convolution,
    min-max convolution,
    dynamic programming,
    approximation algorithms
} 

\category{} 

\relatedversion{} 

\acknowledgements{The author thanks Petru Pascu, Radu Vintan, and Altan Birler for their helpful feedback and the anonymous reviewers whose detailed comments improved the overall presentation.}

\usepackage{amsthm}
\newcommand{\sparagraph}[1]{\vspace{1mm}\noindent {\bf #1}}

\makeatletter
\newcommand{\ovee}{\mathbin{\mathpalette\make@circled\vee}}
\newcommand{\make@circled}[2]{%
  \ooalign{$\m@th#1\smallbigcirc{#1}$\cr\hidewidth$\m@th#1#2$\hidewidth\cr}%
}
\newcommand{\smallbigcirc}[1]{%
  \vcenter{\hbox{\scalebox{0.82}{$\m@th#1\bigcirc$}}}%
}

\makeatletter
\newcommand{\subalign}[1]{%
  \vcenter{%
    \Let@ \restore@math@cr \default@tag
    \baselineskip\fontdimen10 \scriptfont\tw@
    \advance\baselineskip\fontdimen12 \scriptfont\tw@
    \lineskip\thr@@\fontdimen8 \scriptfont\thr@@
    \lineskiplimit\lineskip
    \ialign{\hfil$\m@th\scriptstyle##$&$\m@th\scriptstyle{}##$\hfil\crcr
      #1\crcr
    }%
  }%
}
\makeatother

\begin{document}

\maketitle

\begin{abstract}
Exponential-time approximation has recently gained attention as a practical way to deal with the bitter \textsf{NP}-hardness of well-known optimization problems. We study for the first time the $(1 + \eps)$-approximate min-sum subset convolution. This enables exponential-time $(1 + \eps)$-approximation schemes for problems such as minimum-cost $k$-coloring, the prize-collecting Steiner tree, and many others in computational biology. Technically, we present both a weakly- and strongly-polynomial approximation algorithm for this convolution, running in time $\widetilde O(2^n \log M / \eps)$ and $\widetilde O(2^\frac{3n}{2} / \sqeps)$, respectively. Our work revives research on tropical subset convolutions after nearly two decades.
\end{abstract}

\newpage
\section{Introduction}\label{sec:introduction}

Fast subset convolution is one of the tools in parameterized algorithms which made their way in many of the dynamic programming solutions to well-known \NP-hard problems~\cite{cygan2015algebraic}. Given functions $f$ and $g$ defined on the subset lattice of order $n$, their \emph{sum-product} subset convolution is defined for all $S \subseteq [n] \vcentcolon= \{1, \ldots, n\}$ as
\[
h(S) = (f \ast g)(S) = \displaystyle\sum_{T \subseteq S} f(T)g(S \setminus T).    
\]
Its prominence does not yet come to a surprise since many computationally hard problems accept a convolution-like shape. The remarkable reduction in time-complexity from $O(3^n)$ to $O(2^n n^2)$ by Bj\"orklund et al.~\cite{fsc} represented indeed a breakthrough, allowing these problems to be solved faster. These problems, however, do indeed reduce to a slightly different subset convolution, namely the \emph{min-sum} subset convolution, defined for all $S \subseteq [n]$ as
\begin{equation}
  h(S) = (f \star g)(S) = \min_{T \subseteq S}\big(f(T) + g(S \setminus T)\big).
  \label{eq:min_sum_subset_conv}
\end{equation}
The naïve algorithm inherently takes $O(3^n)$-time, as in the case of any subset convolution: For each set $S$, iterate over all its subsets $T$.\footnote{The running time can be derived from $\sum_{k = 0}^{n} {n \choose k} 2^k = (1 + 2)^n = 3^n$.} In contrast to the speedup achieved by Bj\"orklund et al.~\cite{fsc} for the sum-product subset convolution, the naïve algorithm is still the fastest known for the min-sum subset convolution. This is also the reason why min-sum subset convolution is not used ``as is'', but in a two-step approach: \texttt{(i)} embed the min-sum semi-ring into the sum-product ring, \texttt{(ii)} perform the fast subset convolution in this ring instead. This results in an $\widetilde O(2^n M)$-time algorithm~\cite{fsc}, where $M$ is the largest input value.

This workaround has two limitations: First, the input functions must have a bounded integer range $\{-M, \ldots, M\}$, and second, the running time of the final algorithm -- $\widetilde O(2^n M)$ -- depends on $M$, making the algorithm a pseudo-polynomial one.\footnote{Note that the input size itself is on the order of $O(2^n)$.} This is also where the story behind the min-sum subset convolution ends. Is that all? In what follows, we argue that there is more to it.

\subsection{Approximating Min-Sum Subset Convolution}

Exponential-time approximation algorithms for \NP-hard problems have recently attracted much attention~\cite{exp_approx_bansal2019new,exp_approx_bourgeois2011approximation,exp_approx_escoffier2016super,exp_approx_esmer2022faster, exp_approx_esmer2023approximate, exp_approx_can2024optimally, exp_approx_manurangsi2018mildly}. In the light of this development, we propose the $(1 + \eps)$-\emph{approximate} variant of the min-sum subset convolution as a ``Swiss Army knife'' that enables exponential-time approximation for many \NP-hard optimization problems: Compute the set function $\widetilde h$ such that for all $S \subseteq [n]$ the following holds:
\[
(f \star g)(S) \leq \widetilde h(S) \leq (1 + \eps)(f \star g)(S).
\]
If we could devise a faster-than-naïve algorithm for the approximate counterpart, this would enable $(1 + \eps)$-approximation schemes for a plethora of computationally hard problems that use min-sum subset convolution as a primitive. Typically, these are  problems on graphs such as the minimum Steiner tree problem~\cite{dreyfus1971steiner}, coloring~\cite{convex_recoloring}, and spanning problems on hypergraphs~\cite{spanning_hypergraphs, msth}. On the application side, min-sum subset convolution and its counterpart, \emph{max-sum} subset convolution,\footnote{The operations are performed in the $(\max, +)$-semi-ring.} are also present in computational biology~\cite{protein_network, colorful_subtrees}. This leads us to our driving research question:
\begin{quote}
\itshape
\centering
Are there faster-than-naïve $(1+\eps)$-approximation algorithms for min-sum subset convolution?
\end{quote}
On another note, Bj\"orklund et al.~\cite{fsc} showed that, assuming a bounded input, one can improve on decades-old stagnant results. The most notable of these was the $\widetilde O(2^kn^2M + nm\log M)$-time algorithm for the minimum Steiner tree in graphs with $n$ vertices, $k$ terminals, and $m$ edges with integer weights bounded by $M$. This is indeed the case for all other hard problems that reduce to the min-sum subset convolution: They are all ``scapegoats'' of the bounded-input assumption, having to accept an additional factor of $M$ in their time-complexities. This leads us to our next question:
\begin{quote}
\itshape
\centering
Are there faster-than-naïve $(1 + \eps)$-approximation schemes for convolution-like \NP-hard optimization problems with running time independent of M?
\end{quote}

Next, we present our results.

\subsection{Our results}

\sparagraph{Overview.}~We answer our driving questions in the affirmative. The way we answer them is also intriguing in itself. Let us review it briefly. First, we bring together two lines of research that have been so far considered separately:~\emph{sequence} and \emph{subset} convolutions in semi-rings. This gives us both the weakly- and the strongly-polynomial algorithms for the $(1+\eps)$-approximate min-sum subset convolution. Then, using these algorithms and a standard approximation technique, we obtain $(1 + \eps)$-approximation algorithms for several convolution-like \NP-hard optimization problems. This enriches the current toolbox in exponential-time approximation~\cite{exp_approx_bansal2019new,exp_approx_bourgeois2011approximation,exp_approx_escoffier2016super,exp_approx_esmer2022faster, exp_approx_esmer2023approximate, exp_approx_can2024optimally, exp_approx_manurangsi2018mildly}.

\sparagraph{Min-Sum Subset Convolution.}~Moreover, we revive the line of research on min-sum subset convolution after nearly two decades, by presenting both a weakly- and a strongly-polynomial $(1 + \eps)$-approximation algorithm for this problem: 
\begin{restatable}[]{theorem}{WeaklyPolyMinSumSubsetConvolution}
    \label{thm:weakly_approx_min_sum_subset_convolution}
    $(1+\eps)$-Approximate min-sum subset convolution can be solved in $\widetilde O(2^n \log M / \eps)$-time.
\end{restatable}
\begin{restatable}[]{theorem}{StronglyPolyMinSumSubsetConvolution}
    \label{thm:strongly_approx_min_sum_subset_convolution}
    $(1+\eps)$-Approximate min-sum subset convolution can be solved in $\widetilde O(2^\frac{3n}{2} / \sqeps)$-time.
\end{restatable}%
Note that the strongly-polynomial algorithm (Thm.~\ref{thm:strongly_approx_min_sum_subset_convolution}) outperforms its weakly-polynomial counterpart (Thm.~\ref{thm:weakly_approx_min_sum_subset_convolution}) for huge co-domains.\footnote{Indeed, a sufficient condition is $M = \omega(2^{\sqrt{\eps 2^{n + \delta}}})$, for any $\delta > 0$.} While our weakly-polynomial algorithm uses the scaling technique, as does the fastest algorithm for the $(1 + \eps)$-approximate min-plus \emph{sequence} convolution~\cite{partition_karol}, obtaining the strongly-polynomial algorithm requires a detour in another semi-ring. Namely, in the course of proving it, we had to design an algorithm for a more curious convolution, namely the \emph{min-max} subset convolution, defined for all $S \subseteq [n]$ as
\[
h(S) = (f \ovee g)(S) = \min_{T \subseteq S} \max \{f(T), g(S \setminus T)\}.
\]
To the best of our knowledge, the \emph{min-max} subset convolution has not yet been present in the literature. As a by-product, we show that its exact evaluation does indeed break the natural $O(3^n)$-time barrier:

\begin{table}[htbp]
    \centering
    \begin{tabular}{cccc}
        \toprule
        \textbf{Reference} & \textbf{Type} & \textbf{(Semi-)ring} & \textbf{Running Time} \\
        \midrule
        ad-hoc & exact & any & $O(3^n)$ \\
        Bj\"orklund et al.~\cite{fsc} & exact & $(+, \times)$ & $O(2^n n^2)$ \\
        Bj\"orklund et al.~\cite{fsc} & exact & $(\min, +)$ & $\widetilde O(2^n M)$ \\
        Thm.~\ref{thm:min_max_subset_convolution}, \emph{this work} & exact & $(\min, \max)$ & $\widetilde O(2^\frac{3n}{2})$ \\
        Thm.~\ref{thm:weakly_approx_min_sum_subset_convolution}, \emph{this work} & $(1+\eps)$-approx. & $(\min, +)$ & $\widetilde O(2^n \log M / \eps)$ \\
        Thm.~\ref{thm:strongly_approx_min_sum_subset_convolution}, \emph{this work} & $(1+\eps)$-approx. & $(\min, +)$ & $\widetilde O(2^\frac{3n}{2} / \sqrt{\eps})$ \\
        \bottomrule
    \end{tabular}
    \captionsetup{justification=centering}
    \caption{Reviving research on tropical subset convolutions.}
    \label{tab:line_of_research}
\end{table}
\begin{restatable}[]{theorem}{MinMaxSubsetConvolution}
    \label{thm:min_max_subset_convolution}
    Exact min-max subset convolution can be solved in $\widetilde O(2^\frac{3n}{2})$-time.
\end{restatable}
This generalizes Kosaraju's algorithm~\cite{kosaraju} for the min-max sequence convolution to the subset lattice. We need this intermediate result so that we can instantiate the recent framework of Bringmann, K\"unnemann, and Węgrzycki~\cite{approx_min_plus}, used to obtain the first strongly-polynomial algorithm for the $(1 + \eps)$-approximate min-plus sequence convolution. In particular, they also showed the equivalence between exact min-max sequence convolution and the latter.\footnote{The current version of \cite[Thm.~8.2]{approx_min_plus} contains a typo, relating the min-plus sequence convolution to the min-max \emph{product}. The theorem actually refers to convolution, as can be seen from its proof.} We show that this result naturally holds in the subset setting as well, a missing piece in their work that has immediate application to well-known \NP-hard optimization problems.

\begin{restatable}[Extension of~{\cite[Thm.~1.5]{approx_min_plus}}]{theorem}{Equivalence}
\label{thm:equivalence}
For any $c \geq 1$, if one of the following statements is true, then both are:
\begin{itemize}
    \item $(1 + \eps)$-Approximate \textit{min-sum} subset convolution can be solved in time $\widetilde O(2^{cn} / \poly(\eps))$,
    \item Exact \textit{min-max} subset convolution can be solved in time $\widetilde O(2^{cn})$.\\
\end{itemize}
\end{restatable}
Simply using the main backbone of their framework, the Sum-to-Max-Covering lemma~\cite[Thm.~1.7]{approx_min_plus} (we define it later in Sec.~\ref{sec:approx_min_sum_subset_conv}), one could obtain an $\widetilde O(2^\frac{3n}{2} / \eps)$-time algorithm for the approximate min-sum subset convolution; note that the running time is independent of $M$. However, the dependence on the parameter $\eps$ can indeed be improved using a refined analysis. In Table~\ref{tab:line_of_research}, we summarize our results, which revive the research on tropical subset convolutions after nearly two decades and enable out-of-the-box exponential-time $(1 + \eps)$-approximation schemes for several hard optimization problems, as outlined in the following.

\sparagraph{Applications.}~Min-sum subset convolution is present in the dynamic programming formulation of many \NP-hard optimization problems. Thus, we obtain $(1 + \eps)$-approximation schemes for these problems as well, such as the minimum-cost $k$-coloring~\cite{Cygan_book}, the prize-collecting Steiner tree, and two applications in computational biology: the maximum colorful subtree problem~\cite{colorful_subtrees}, which instead requires the \emph{max-sum} subset convolution -- for which we devise an $(1 - \eps)$-approximation scheme -- and a problem on protein interaction networks~\cite{protein_network} (already proposed as an application by Bj\"orklund et al.~\cite{fsc}). We outline two of these results here:

\begin{theorem}
\label{thm:min_cost_k_coloring}
$(1 + \eps)$-Approximate minimum-cost $k$-coloring can be solved in $\widetilde O(2^\frac{3n}{2} / \sqeps)$-time.
\end{theorem}

\begin{theorem}
\label{thm:prize_collecting_steiner_tree}
$(1 + \eps)$-Approximate prize-collecting Steiner tree can be solved in $\widetilde O(2^\frac{3s^+}{2} / \sqeps)$-time.\footnote{The parameter $s^+$ is the number of proper potential terminals; see Appendix~\ref{appendix:prize_collecting} for its definition.}
\end{theorem}

\subsection{Related Work}

Our work builds a new bridge from sequence convolutions to subset convolutions. The sequence convolution of two sequences $a$ and $b$ in the $(\oplus, \otimes)$-semi-ring is computed as $c_k=\oplus_{i + j = k}\:(a_i \otimes b_j)$, while the subset convolution of two set functions $f$ and $g$ in the same semi-ring is defined as $h(S) = \oplus_{T\subseteq S}\big(f(T) \otimes g(S\setminus T)\big)$. Independently on the structure of the convolution, i.e., either on sequences or on subsets, there are three types of sequence convolution which predominate in the literature, corresponding to three different (semi-)rings: (a) $(+, \times)$-ring, (b) $(\min, \max)$-semi-ring, and (c) $(\min, +)$-semi-ring, respectively. In the following, we outline previous work for both convolution types.

\sparagraph{Sequence Convolution.}
\begin{enumerate}[(a)]
    \item The sequence convolution in the $(+, \times)$-ring can be solved in time $O(n \log n)$ via FFT.
    \item In the $(\min, \max)$-semi-ring, Kosaraju presented an $\widetilde O(n \sqrt n)$-time algorithm~\cite{kosaraju}, and even conjectured that his algorithm can be improved to $\widetilde{O}(n)$. However, no improvement has been reported so far~\cite{approx_min_plus}.
    \item Whether Min-Plus Sequence Convolution can be computed in time $O(n^{2 - \delta})$ for any $\delta > 0$ is still an open problem~\cite{cygan_conjecture_1, kuennnemann_conjecture_2}. The fastest algorithm to date runs in time $n^2 / 2^{\Omega(\sqrt{\log n})}$, by combining the reduction to Min-Plus Matrix Product by Bremner et al.~\cite{necklaces} and an algorithm for the latter due to Williams~\cite{williams_first}, subsequently derandomized by Chan and Williams~\cite{chan_williams}. In case the values are bounded by a constant $W$, the convolution can be performed in $\widetilde O(n W)$-time~\cite[Lemma 5.7.2]{karol_phd}. The first $(1 + \eps)$-approximation algorithm has been presented by Backurs et al.~\cite{backurs_tree_sparsity}, as an application for the Tree Sparsity problem. Their algorithm runs in time $O(\frac{n}{\eps^2} \log n \log^2 W)$ and has been improved to $O(\frac{n}{\eps} \log (n / \eps) \log W)$~\cite{partition_karol}. Bringmann et al.~\cite{approx_min_plus} finally provided a strongly-polynomial algorithm, i.e., independent of $W$, in time $\widetilde O(n^{3/2} / \sqeps)$. Their result is indeed more general, obtaining a framework that can also be applied to All-Pairs Shortest Pairs (APSP). In this paper, we extend their framework to the context of subset convolutions, with immediate application in $(1 + \eps)$-approximation schemes of several convolution-like \NP-hard problems.
\end{enumerate}

\sparagraph{Subset Convolution.}
\begin{enumerate}[(a)]
    \item Prior to the work of Bj\"orklund et al.~\cite{fsc}, the best algorithm for any subset convolution was the straightforward $O(3^n)$-time evaluation. They provided an $O(2^n n^2)$-time algorithm for the subset convolution in the $(+, \times)$-ring, by relating the FFT to the M\"obius transform on the subset lattice. This result found applications in many problems, among them the Domatic Number and the Chromatic Number (and many others) in $2^n n^{O(1)}$-time~\cite{set_partitioning}.
    \item As far as the $(\min, \max)$-semi-ring is concerned, there is, to our best knowledge, no algorithm that runs in time better than $O(3^n)$. In this paper, we provide such an algorithm, as a by-product.
    \item On the other hand, the convolution in the $(\min, +)$-semi-ring received more attention, as it is implicitly present in many hard optimization problems, a prominent example being the minimum Steiner tree. While so far no $o(3^n)$-time algorithm is known to exist for this convolution, a simple embedding technique can leverage the speedup obtained for the sum-product subset convolution. The caveat is that the input values are required to be bounded by a constant $M$. In this case, the convolution runs in time $\widetilde{O}(2^n M)$~\cite{fsc} (compare to the $\widetilde{O}(n W)$-time algorithm for the corresponding sequence convolution). In particular, no $(1 + \eps)$-approximation algorithm has been known until our work.
\end{enumerate}

\noindent We provide in Appendix~\ref{appendix:convolutions} the definitions of the convolutions we will be discussing.

\subsection{Organization}

We organize the paper as follows:
We begin with the exact algorithm for the \emph{min-max} subset convolution (Sec.~\ref{sec:exact_min_max_subset_conv}).
Subsequently, in Sec.~\ref{sec:approx_min_sum_subset_conv}, we introduce our weakly-polynomial algorithm and then present a simple strongly-polynomial approximation algorithm for the min-sum subset convolution. This serves the equivalence between the exact \emph{min-max} subset convolution and the approximate \emph{min-sum} subset convolution, using the simple algorithm for the former.
Then, we outline an improved strongly-polynomial approximation scheme for the \emph{min-sum} subset convolution, adapting that of Bringmann et al.~\cite{approx_min_plus} to the subset setting.
Based on the improved variant, we introduce in Sec.~\ref{sec:applications} the first $(1+\eps)$-approximation scheme for the minimum cost $k$-coloring problem (we outline many other applications in Appendix~\ref{appendix:further_apps}).
We conclude in Sec.~\ref{sec:conclusion} with several cases that have been considered for the min-plus sequence convolution, but are completely missing in the subset setting. 

\section{Preliminaries}

\sparagraph{Notation.} We use the $\widetilde O$-notation to suppress poly-logarithmic factors in the input size and in $\eps$, but never in $M$ (or in $W$ in the sequence setting). We denote by $[n]$ the set $\{1, \ldots, n\}$ for $n$ a natural number. To be consistent with the notation used in previous work in \emph{both} research fields, we refer to the $(\min, +)$-semi-ring as \emph{min-sum} for subset convolutions and as \emph{min-plus} for sequence convolutions. The same distinction is done between the terms $M$ and $W$, corresponding to the largest (finite) input value in the subset and sequence setting, respectively. We refer to the $(+, \times)$-ring as the sum-product ring, and to the $(\min, \max)$-semi-ring as the min-max semi-ring. We will later need Iverson's bracket notation: Given a property $P$, $[P]$ takes the value $1$ whenever $P$ is true and $0$ otherwise.

\sparagraph{Machine Model and Input Format.} We follow the same setup as assumed by Bringmann et al.~\cite{approx_min_plus}, namely: The input numbers in \emph{approximate} problems are represented in floating-point, whereas those in \emph{exact} problems are integers in the usual bit representation. This particular setup is needed when extending the equivalence between exact min-max and approximate min-plus to the subset setting. To economically use the space for the actual applications, we postpone the details on the setup to our Appendix~\ref{appendix:input_format}.

\section{Exact Min-Max Subset Convolution}\label{sec:exact_min_max_subset_conv}

Both our strongly-polynomial algorithms for evaluating the $(1 + \eps)$-approximate \emph{min-sum} subset convolution heavily rely on the \emph{min-max} subset convolution, defined for all $S \subseteq [n]$ as
\begin{equation}
h(S) = (f \ovee g)(S) = \min_{T \subseteq S} \max \{f(T), g(S \setminus T)\}.
\label{eq:min_max_sc}
\end{equation}
In this section, we provide an $\widetilde O(2^{\frac{3n}{2}})$-time algorithm for Eq.~\eqref{eq:min_max_sc}, inspired by Kosaraju's $\widetilde O(n \sqrt{n})$-time algorithm for min-max \emph{sequence} convolution~\cite{kosaraju}.
We note that similar techniques to Kosaraju's have been considered for the min-max matrix product by Shapira et al.~\cite{shapira_apbp}, Williams et al.~\cite{williams_apbp}, and Duan and Pettie~\cite{duan_pettie_apbp}, as applications to the all-pairs bottleneck pairs problem.
\MinMaxSubsetConvolution*
\begin{proof}
We closely follow Kosaraju's algorithm for the min-max \emph{sequence} convolution~\cite{kosaraju}: First, we collect all values of $f$ and $g$ in a list $\mathcal{L} \vcentcolon= \{(f(S), S), (g(S), S) \mid S \subseteq [n]\}$ and sort it by its first argument (ties are broken arbitrarily). We then divide the sorted list into $O(\sqrt{2^n})$ chunks. Let $\mathcal{C}_i$ be the current chunk, and $\mathcal{C}_i^{1}$ and $\mathcal{C}_i^{2}$ be its projections onto the first and second arguments, respectively. We apply fast (boolean) subset convolution on $[f \leq \max \mathcal{C}_i^{1}]$ and $[g \leq \max \mathcal{C}_i^{1}]$ and obtain $\hat h = [f \leq \max \mathcal{C}_i^{1}] \ast [g \leq \max \mathcal{C}_i^{1}]$. In particular, $\hat h \equiv [h \leq \max \mathcal{C}_i^{1}]$, and thus a non-zero $\hat h(S)$ tells us whether the actual $h(S)$ is bounded above by $\max \mathcal{C}_i^{1}$. At this point, we only need to figure out what the \emph{actual} value of $h(S)$ is.

To this end, we consider only those sets $S$ for which $\mathcal{C}_i$ is the \emph{first} chunk that made $\hat h(S)$ positive. Let us call this set $\mathcal{S}_i$. We thus consider each $S \in \mathcal{S}_i$ separately and compute its exact $h(S)$ value by iterating over the $O(\sqrt{2^n})$ sets $U \in \mathcal{C}_i^{2}$. Formally, we set
\begin{equation}
h(S) = \min_{\subalign{U &\in \mathcal{C}_i^{2} \\ U &\subseteq S}} \max\:\{f(U), g(S \setminus U)\}.
\label{eq:local_opt}
\end{equation}
The correctness follows from construction. Let us analyze the running time. Sorting the list takes $\widetilde O(2^n)$-time. Then, it takes $O(|\mathcal{S}_i|\sqrt{2^n})$-time to evaluate Eq.~\eqref{eq:local_opt} for each chunk $\mathcal{C}_i$, since $|\mathcal{C}_i| = O(\sqrt{2^n})$ by construction. Since we run $O(\sqrt{2^n})$ convolutions in the sum-product ring and Eq.~\eqref{eq:local_opt} is run only once per $S \subseteq [n]$, the total running time reads $O(2^n n^2 \sqrt{2^n} + 2^n \sqrt{2^n}) = \widetilde O(2^\frac{3n}{2})$.
\end{proof}

Notably, the running time given in Thm. 3 is faster than $O(3^n)$. This result will serve as the basis for our \emph{strongly-polynomial} approximation algorithms for \emph{min-sum} subset convolution.

\section{Approximate Min-Sum Subset Convolution}\label{sec:approx_min_sum_subset_conv}

We now turn to our original goal: \emph{approximating the min-sum subset convolution}. We present two ways to do this, either via a weakly-polynomial time algorithm or a strongly-polynomial time algorithm. Let us start with the former.

\subsection{Weakly-Polynomial Algorithm}

\begin{restatable}[]{algorithm}{WeaklyPolyAlgorithm}
    \caption{$\textsc{ApxMinSumSubsetConv}(f,g,\eps)$ [Thm.~\ref{thm:weakly_approx_min_sum_subset_convolution}] -- Running time: $\widetilde O(2^n \log M / \eps)$}
	\label{algo:weakly_polynomial_conv}
\begin{algorithmic}[1]
    \Procedure{Scale}{$f, q, \eps$}
    	\State \Return $S \mapsto
            \begin{cases}\big\lceil \frac{2f(S)}{\eps q}  \big\rceil, & \text{if} \; \big\lceil \frac{2f(S)}{\eps q} \big\rceil \leq \big\lceil \frac{4}{\eps} \big\rceil,\\
    		\infty, & \text{otherwise.}\\
    	\end{cases}$
    \EndProcedure
    \State
    \State Set $\widetilde{h}(S) \gets \infty$ for all $S \subseteq [n]$
    \For {$q = 2^{\ceil{\log{2M}}}, \ldots, 4, 2, 1$}
      \State $f_q, g_q \gets \textsc{Scale}(f,q,\eps), \textsc{Scale}(g,q,\eps)$
      \State $h_q \gets \textsc{ExactMinSumSubsetConv}(f_q, g_q)$ 
      \State $\widetilde{h}(S) \gets \min \{\widetilde{h}(S), h_q(S) \cdot \frac{\eps q}{2} \}$ for all $S \subseteq [n]$
    \EndFor
    \State \Return $\widetilde{h}$
\end{algorithmic}
\end{restatable}
Our weakly-polynomial algorithm is based on the scaling technique, which has been successfully used in several graph problems~\cite{goldberg_graph, gabow_graph, orlin_graph, tarjan_graph, matching_graph}.
In the context of sequence convolution, it has been used by Backurs et al.~\cite{backurs_tree_sparsity} to provide the first $(1+\eps)$-approximation scheme for the min-plus sequence convolution, running in time $O(\frac{n}{\eps^2} \log n \log^2 W)$.
This has been later improved by Mucha et al.~\cite{partition_karol} to $O(\frac{n}{\eps} \log (n / \eps) \log W)$-time.

We can use a method similar to that in Ref.~\cite{partition_karol} and solve the $(1+\eps)$-approximate min-sum subset convolution in weakly-polynomial time. The key insight is that we can replace the ``heart'' of their algorithm, namely using the exact $\widetilde O(n W)$-time min-plus sequence convolution for bounded input, with the subset counterpart running in time $\widetilde O(2^n M)$, due to Bj\"orklund et al.~\cite{fsc}. To our best knowledge, this observation had not appeared in the literature before. In particular, the applications themselves had no $(1 + \eps)$-approximation schemes prior to our work. We outline our algorithm in Alg.~\ref{algo:weakly_polynomial_conv} and prove Thm.~\ref{thm:weakly_approx_min_sum_subset_convolution} in Appendix~\ref{appendix:weakly_poly}.

\WeaklyPolyMinSumSubsetConvolution*

\sparagraph{Striving for Strongly-Polynomial Time.} While the weakly-polynomial algorithm allows us to reduce the linear dependence on $M$ in the running time of the algorithm of Bj\"orklund et al.~\cite{fsc} to a logarithmic one, the dependence on $M$ still remains.
The caveat is that this dependence will carry over in the actual applications. \emph{Is this the best we can hope for?} In the following, we will show that we can \emph{completely} discard this dependence.

In a recent breakthrough, Bringmann et al.~\cite{approx_min_plus} introduced the first strongly-polynomial algorithm for evaluating the approximate min-plus \emph{sequence} convolution: Given two sequences $a$ and $b$, their min-plus convolution\footnote{In this work, since we are dealing with two types of convolution, on sequences and subsets, respectively, we will consistently add the specification ``sequence'' to avoid any confusion.} is a sequence $c$, where $c_k = \min_{i \leq k}\:(a_i + b_{k - i})$.
As in our case, its approximate variant asks for a sequence $\widetilde c$ such that $c_k \leq \widetilde c_k \leq (1 + \eps)c_k$. Notoriously, previous work on this problem had employed the scaling trick~\cite{zwick_apsp}, which relies on the largest input value $W$ and introduces an additional $\log W$ into the running time.\footnote{As noted above, the term $W$ in the literature on sequence convolutions corresponds to $M$ in subset convolutions.} The approximate convolution finds a natural application in the Tree Sparsity problem~\cite{tree_sparsity_initial}, where its exact algorithm is prohibitively expensive.

The aforementioned authors asked a fundamental question: \emph{Is it possible to completely avoid the scaling trick?}
They indeed answered this question affirmatively, by designing the first \emph{strongly-polynomial} $(1 + \eps)$-approximation schemes for the all-pairs shortest pairs (APSP) problem and many others, including the min-plus sequence convolution. Their cornerstone result is the Sum-to-Max-Covering lemma:

\begin{lemma}[Sum-to-Max Covering~\cite{approx_min_plus}]
\label{lemma:sum_to_max_covering}
Given vectors $A, B \in \mathbb{R}^d_{+}$ and a parameter $\eps > 0$, there are vectors $A^{(1)}, \ldots, A^{(s)}, B^{(1)}, \ldots, B^{(s)} \in \mathbb{R}^d_{+}$ with $s = O(\frac{1}{\eps}\log\frac{1}{\eps} + \log d \log \frac{1}{\eps})$ such that for all $i, j \in [d]$:
\[
    A[i] + B[j] \leq \displaystyle\min_{\ell\in[s]}\max\{A^{(\ell)}[i], B^{(\ell)}[j]\} \leq (1 + \eps)(A[i] + B[j]).
\]
The vectors $A^{(1)}, \ldots, A^{(s)}, B^{(1)}, \ldots, B^{(s)}$ can be computed in time $O(\frac{d}{\eps} \log\frac{1}{\eps} + d \log d\log \frac{1}{\eps})$.
\end{lemma}
The motivation behind their lemma is that min-max sequence convolution can be solved in $\widetilde O(n \sqrt{n})$-time (using Kosaraju's algorithm~\cite{kosaraju}), in contrast to the best-known algorithm for min-plus convolution which runs in $O(n^2)$-time.\footnote{Indeed, it is an open problem whether this is also the best that can be achieved~\cite{cygan_conjecture_1, kuennnemann_conjecture_2}.} In this light, one can first compute the auxiliary vectors $A^{(1)}, \ldots, A^{(s)}, B^{(1)}, \ldots, B^{(s)}$ via Lemma~\ref{lemma:sum_to_max_covering} and then solve the $(1 + \eps)$-approximate min-plus sequence convolution by performing $s$ min-max sequence convolutions.

\subsection{Simple Strongly-Polynomial Approximation Algorithm}\label{subsec:inter}

The reader can already see the applicability to our setting: We can apply the Sum-to-Max Covering lemma in min-sum \emph{subset} convolution. Indeed, this is the cornerstone neglected by Bringmann et al.~\cite{approx_min_plus}. This key insight is the building block behind our new results in Sec.~\ref{sec:applications}.  

\begin{algorithm}
    \caption{$\textsc{ApxMinSumSubsetConv}(f,g,\eps)$ [Lemma~\ref{lemma:simple_approx_min_sum_subset_conv}] -- Running time: $\widetilde O(2^\frac{3n}{2} / \eps)$}
	\label{alg:approx-minconv-simple}
\begin{algorithmic}[1]
    \State $\{f^{(1)},\ldots,f^{(s)},g^{(1)},\ldots,g^{(s)}\} \gets \textsc{SumToMaxCovering}(f, g, \eps)$ (Lemma~\ref{lemma:sum_to_max_covering})
    \State $h^{(\ell)} \gets \textsc{MinMaxSubsetConv}(f^{(\ell)}, g^{(\ell)})$ for each $\ell \in [s]$
    \State $\widetilde h(S) \gets \displaystyle\min_{\ell \in [s]}\:\{ h^{(\ell)}[\texttt{idx}(S)] \}$ for each $S \subseteq [n]$
    \State \Return $\widetilde h$
\label{algo:approx_min_sum_subset_conv}
\end{algorithmic}
\end{algorithm}

\begin{restatable}[]{lemma}{SimpleStrongly}
\label{lemma:simple_approx_min_sum_subset_conv}
$(1+\eps)$-Approximate min-sum subset convolution can be solved in time $\widetilde O(2^\frac{3n}{2}/\eps)$.
\end{restatable}

We prove Lemma~\ref{lemma:simple_approx_min_sum_subset_conv} in Appendix~\ref{appendix:simple_strongly}. Note that this \emph{is still an intermediate result}. In the following, we will use it to prove the equivalence between exact min-max subset convolution and approximate min-sum subset convolution. This represents the extension of the same result at the level of sequence convolutions. We will prove Thm.~\ref{thm:equivalence} in Appendix~\ref{appendix:equivalence}.
\Equivalence*

\subsection{Improved Strongly-Polynomial Approximation Algorithm}
We used the $\widetilde O(2^\frac{3n}{2} / \eps)$-time algorithm for the approximate min-sum subset convolution (Lemma~\ref{lemma:simple_approx_min_sum_subset_conv}) as a basis for the equivalence between \emph{approximate} min-sum subset convolution and \emph{exact} min-max subset convolution. In the following, we design an algorithm with an improved running time of $\widetilde O(2^\frac{3n}{2} / \sqeps)$. Our key insight is that the new algorithm can be analyzed using a toolbox similar to that of Bringmann et al.~\cite{approx_min_plus}. This will be the running time we use for the applications in Sec.~\ref{sec:applications} and Appendix~\ref{appendix:further_apps}.

We split the main result in the corresponding subcases, namely \texttt{(i)} \emph{distant} summands and \texttt{(ii)} \emph{close} summands. The main theorem, Thm.~\ref{thm:strongly_approx_min_sum_subset_convolution}, will follow from these.

\begin{restatable}[]{lemma}{DistantSummandsLemma}
    \label{lemma:distant_summands}
    Given set functions $f, g$ on the subset lattice of order $n$ and a parameter $\eps > 0$, let $h = f \star g$ be their min-sum subset convolution. We can compute in time $O(2^\frac{3n}{2} \polylog(\frac{2^n}{\eps}))$ a set function $\widetilde h$ such that for any $S \subseteq [n]$ we have
    \begin{enumerate}[(i)]
        \item $h(S) \leq \widetilde h(S)$, and
        \item if there is $T \subseteq S$ with $h(S) = f(T) + g(S \setminus T)$ and $\frac{f(T)}{g(S \setminus T)} \not\in [\frac{\eps}{4}, \frac{4}{\eps}]$, then $\widetilde h(S) \leq (1+\eps)h(S)$.
    \end{enumerate}
\end{restatable}

In Appendix~\ref{appendix:distant_covering_lemma}, we show that Alg.~\ref{algo:distant_conv} returns such an $\widetilde h$. Briefly, we view the set functions $f, g$ as vectors in $\mathbb{R}^{2^n}_{+}$.
Hence, we can use Distant Covering from the framework~\cite[Cor.~5.10]{approx_min_plus} on $f, g$ with $\eps' \vcentcolon= \frac{\eps}{4}$ and obtain $f^{(1)}, \ldots, f^{(s)}, g^{(1)}, \ldots, g^{(s)}$ with $s = O\left(\polylog(\frac{2^n}{\eps})\right)$.\footnote{For completeness, we provide \cite[Cor.~5.10]{approx_min_plus} as Lemma~\ref{lemma:distant_covering_lemma} in Appendix~\ref{appendix:distant_covering_lemma}.}
For each $l \in [s]$, we compute $h^{(l)} = f^{(l)} \ovee g^{(l)}$ and scale the entry-wise minimum by $\frac{1}{1 - 2\eps'}$.

The case of \emph{close} summands is covered by the following result:

\begin{restatable}[]{lemma}{CloseSummandsLemma}
    \label{lemma:close_summands}
    Given set functions $f, g$ on the subset lattice of order $n$ and a parameter $\eps > 0$, let $h = f \star g$ be their min-sum subset convolution. We can compute in time $\widetilde O(2^\frac{3n}{2} / \sqeps)$ a set function $\widetilde h$ such that for any $S \subseteq [n]$ we have
    \begin{enumerate}[(i)]
        \item $h(S) \leq \widetilde h(S)$, and
        \item if $T \subseteq S$ with $h(S) = f(T) + g(S \setminus T)$ and $\frac{f(T)}{g(S \setminus T)} \in [\frac{\eps}{4}, \frac{4}{\eps}]$, then $\widetilde h(S) \leq (1+\eps)h(S)$.
    \end{enumerate}
\end{restatable}
In Appendix~\ref{appendix:close_summands_lemma}, we show that Alg.~\ref{algo:close_conv} returns such an $\widetilde h$. We point out one aspect that a careful reader will notice, namely that the for-loop at line 6 (Alg.~\ref{algo:close_conv}) does indeed use the largest value $M$. At first glance, this seems to contradict our assumption that we are trying to avoid running times that include $M$. The key insight to understanding this, used by Bringmann et al.~\cite{approx_min_plus} in the sequence setting, is the fact that $q$ grows geometrically, and thus the number of entries of $f_q$ and $g_q$ that are not set to $\infty$ by \textsc{Scale} is bounded by $O(2^n \log \reveps)$, as we analyze in the proof of Lemma~\ref{lemma:close_summands}. We finally conclude with the final result:
\StronglyPolyMinSumSubsetConvolution*
\begin{proof}
We run both algorithms for the distant and close summands, respectively, and take the entry-wise minimum (Alg.~\ref{algo:apx_min_sum_subset_conv_final} in Appendix~\ref{appendix:algo_extra} shows this). Correctness and running time follow from those of Lemmas~\ref{lemma:distant_summands} and~\ref{lemma:close_summands}.
\end{proof}
\begin{restatable}[]{corollary}{MaxSumSubsetConvolution}
    $(1 - \eps)$-Approximate max-sum subset convolution can be solved in time $\widetilde O(2^\frac{3n}{2} / \sqeps)$.
\end{restatable}

\section{Applications}\label{sec:applications}

Armed with the novel strongly-polynomial algorithm for $(1 + \eps)$-approximate min-sum subset convolution running in time $\widetilde O(2^\frac{3n}{2} / \sqeps)$, we can now develop $(1 + \eps)$-approximation schemes\footnote{Or $(1 - \eps)$-approximation schemes where the \emph{max-sum} subset convolution is applicable.} for a plethora of convolution-like \NP-hard optimization problems. In particular, we target all problems that are ``scapegoats'' of the approach employing min-sum subset convolution on bounded input. In other words, these are all exact algorithms with a dependence on $M$ in their running time. Our goal is to enable out-of-the-box $(1 + \eps)$-approximation schemes in time $\widetilde O(2^\frac{3p}{2} / \sqeps)$, where $p$ is problem-specific (as a rule of thumb, it is always the exponent in the running time $O(3^p)$ of the exact evaluation). The following problem is intended to demonstrate this technique.

\subsection{Minimum-Cost $k$-Coloring}

In their book on parameterized algorithms, Cygan et al.~\cite{Cygan_book} propose a variant of $k$-coloring, which they entitle minimum-cost $k$-coloring, and devise an $\widetilde O(2^n M)$-time algorithm for it~\cite[Thm.~10.18]{cygan2015algebraic}. We are given an undirected graph $G$, an integer $k$, and a cost function $c : V(G) \times [k] \to \{-M, \ldots, M\}$. The cost of a coloring $\chi : V(G) \to [k]$ is defined as $\displaystyle\sum_{v\in V(G)} c(v, \chi(v))$, i.e., coloring vertex $v$ with color $i$ incurs cost $c(v, i)$. The task is to determine the minimum cost of a $k$-coloring of $G$ (if such a coloring exists). To this end, it is handy to introduce a function $s_i : 2^{V(G)} \to \mathbb{Z} \cup \{+\infty\}$ such that for every $X \subseteq V(G)$ we have
\begin{equation*}
s_i(X) =
\begin{cases}
    \displaystyle\sum_{x\in X} c(x, i), & \text{if } X \text{ is an independent set,}\\
    +\infty,  & \text{otherwise.}
\end{cases}
\label{eq:coloring_mapping}
\end{equation*}
Then, one can compute the minimum cost of a $k$-coloring of $G[X]$ as $(s_1 \star \ldots \star s_k)(X)$.
This reduces to simply performing $k - 1$ min-sum subset convolutions. To this end, their proposed algorithm runs in $\widetilde O(2^n M)$-time, by applying the standard min-sum subset convolution for bounded input. We show how to obtain an $(1 + \eps)$-approximation in $\widetilde O(2^\frac{3n}{2} / \sqeps)$-time; the running time is independent of $M$:
\begin{theorem}
    \label{thm:min_cost_coloring_inter}
    If $(1 + \eps)$-Approximate min-sum subset convolution runs in $T(n, \eps)$-time, then an $(1 + \eps)$-approximate minimum-cost $k$-coloring can be found in time $O(T(n, \frac{\eps}{k-1}))$.
\end{theorem}
\begin{proof}
Consider the evaluation of the min-sum subset convolution between two set functions $f$ and $g$ at each step. Setting a relative error $\delta > 0$ for each convolution call, we obtain a cumulative relative error bounded by $(1 + \delta)^{k - 1}$. By setting $\delta = \Theta(\frac{\eps}{k - 1})$, we obtain a relative error of at most $\eps$.
\end{proof}

As a corollary, Thm.~\ref{thm:min_cost_coloring_inter} implies Thm.~\ref{thm:min_cost_k_coloring}. Referring to Ref.~\cite{fomin_ref_10, fomin_ref_13}, Fomin et al.~\cite{fomin_exp_algos} point out that if certain reasonable complexity conjectures hold, then $k$-coloring itself is hard to approximate within $n^{1 - \epsilon}$, for any $\epsilon > 0$. It is interesting to ask whether our time bound is optimal for the $(1 + \eps)$-approximation. We leave this as future work.

\sparagraph{Other Applications.} We provide many other applications, such as the prize-collecting Steiner tree and two other applications in computational biology, in Appendix~\ref{appendix:further_apps}. Note that in all applications, we can always replace the strongly-polynomial approximation algorithm (Thm.~\ref{thm:strongly_approx_min_sum_subset_convolution}) for the min-sum subset convolution with the weakly-polynomial one (Thm.~\ref{thm:weakly_approx_min_sum_subset_convolution}) and alternatively obtain $\widetilde O(2^p \log M / \eps)$-time approximation schemes, where $p$ is the problem-specific parameter, e.g., $p$ is equal to $n$ in the minimum-cost $k$-coloring problem.

\section{Discussion}\label{sec:conclusion}

There seemed to be an unyielding ``isthmus'' between the results on sequence convolutions and subset convolutions on semi-rings. In the following, we outline several future work directions, inspired by research on the sequence setting.

\sparagraph{Polynomial Speedups.}~Currently, the fastest known algorithm for min-plus sequence convolution runs in time $n^2 / 2^{\Omega(\sqrt{\log n})}$, by combining the reduction to Min-Plus Matrix Product by Bremner et al.~\cite{necklaces} and an algorithm for the latter due to Williams~\cite{williams_first}, subsequently derandomized by Chan and Williams~\cite{chan_williams}. This has been the culmination of a long line of research starting with the $O(n^2 / \log n)$-time algorithm due to Bremner et al.~\cite{necklaces}. This leads us to ask whether such speedups can be generalized to the subset context as well:
\begin{quote}
\itshape
\centering
Are there polynomial-factor speedups for the min-sum subset convolution?
\end{quote}
In particular, we are not aware of any $O(3^n / n)$-time algorithm.

\sparagraph{Min-Sum Subset Convolution Conjecture.}~The lack of faster algorithms for this problem leads us to conjecture that a similar scenario as in the case of the min-plus sequence convolution~\cite{cygan_conjecture_1, kuennnemann_conjecture_2} is also present in the subset setting:
\begin{conjecture}\label{conj:min_sum_subset_conv}
There is no $O((3 - \delta)^{n} \polylog(M))$-time exact algorithm for min-sum subset convolution, with $\delta > 0$.
\end{conjecture}
Indeed, an interesting future work is to find out whether both conjectures are equivalent.

\sparagraph{Exploiting Kernelization.}~Certain \NP-hard problems accept polynomial size kernels via the well-known Frank-Tardos' framework~\cite{frank_tardos}. In Appendix~\ref{appendix:lossy_kernels}, we discuss it in the context of our work and leave as future work the possibility of using lossy kernels~\cite{lossy_kernels} for the applications we treated in this paper.



\bibliography{approx_min_sum}

\newpage

\appendix

\section{Convolution Definitions}\label{appendix:convolutions}

\defproblem{Min-Plus Sequence Convolution}
{Sequences $(a[i])_{i=1}^{n},\, (b[j])_{j=1}^{n} \in \mathbb{R}^n_+$} 
{Compute sequence $(c[k])_{k=1}^{n}$ with $a[k] = \displaystyle\min_{i+j=k} (a[i]+b[j])$}

\defproblem{Min-Max Sequence Convolution}
{Sequences $(a[i])_{i=1}^{n},\, (b[j])_{j=1}^{n} \in \mathbb{R}^n_+$} 
{Compute sequence $(c[k])_{k=1}^{n}$ with $c[k] = \displaystyle\min_{i+j=k}\max\{a[i], b[j]\}$}

\defproblem{Min-Sum Subset Convolution}
{Set functions $f, g$ with $f(S), g(S) \in \mathbb{R}^n_+$, for $S \subseteq [n]$} 
{Compute set function $h$ with $h(S) = \displaystyle\min_{T\subseteq S}\left(f(T) + g(S\setminus T)\right)$}

\defproblem{Min-Max Subset Convolution}
{Set functions $f, g$ with $f(S), g(S) \in \mathbb{R}^n_+$, for $S \subseteq [n]$} 
{Compute set function $h$ with $h(S) = \displaystyle\min_{T\subseteq S} \max\{f(T), g(S \setminus T)\}$}

\defproblem{$(1+\eps)$-Approximate Min-Plus Sequence Convolution}
{Sequences $(a[i])_{i=1}^{n},\, (b[j])_{j=1}^{n} \in \mathbb{R}^n_+$}
{Compute a sequence $(\widetilde c[k])_{k=1}^{n}$ with $c[k] \le \widetilde c[k] \le (1+\eps) c[k]$ for any $k \in [n]$, where $c$ denotes the output of Min-Plus Sequence Convolution}

\defproblem{$(1+\eps)$-Approximate Min-Sum Subset Convolution}
{Set functions $f, g$ with $f(S), g(S) \in \mathbb{R}^n_+$, for $S \subseteq [n]$}
{Compute a set function $\widetilde h$ with $h(S) \le \widetilde h(S) \le (1+\eps) h(S)$ for any $S \subseteq [n]$, where $h$ denotes the output of Min-Sum Subset Convolution}

\section{Weakly-Polynomial Approximation Algorithm}\label{appendix:weakly_poly}

We prove the correctness and analyze the running time of our weakly-polynomial approximation algorithm, Alg.~\ref{algo:weakly_polynomial_conv}, for the min-sum subset convolution which uses the scaling technique. It is an adaptation of the fastest scaling-based approximation algorithm by Mucha et al.~\cite{partition_karol} for the min-plus sequence convolution. To this end, we will use the following lemma, which has been used by Mucha et al.~\cite{partition_karol} (see their Lemma 6.2), being inspired by \cite[Lemma 5.1]{zwick_lemma} and \cite[Lemma 1]{ohad_lemma}:
\begin{lemma}[{\cite[Lemma 6.2]{partition_karol}}]
\label{lemma:weakly_helper}
For natural numbers $x,y$, and positive $q,\eps$ satisfying $q \le x+y$
and $0 < \eps < 1$ it holds:
\begin{eqnarray*}
    x+y &\le \Big(\Big\lceil \frac{2x}{q\epsilon} \Big\rceil +
    \Big\lceil \frac{2y}{q\epsilon} \Big\rceil\Big)\frac{q\epsilon}{2} &< (x+y)(1+\epsilon), \\
    (x+y)(1-\epsilon) &< \Big(\Big\lfloor \frac{2x}{q\epsilon} \Big\rfloor +
    \Big\lfloor \frac{2y}{q\epsilon} \Big\rfloor\Big)\frac{q\epsilon}{2} &\le x+y.
\end{eqnarray*}
\end{lemma}

\WeaklyPolyMinSumSubsetConvolution*
\begin{proof}
To prove the theorem, we can employ similar techniques to the those introduced by Mucha et al.~\cite{partition_karol}. The key insight is to replace the exact min-plus sequence convolution for bounded input with the its subset counterpart.

Let us consider Alg.~\ref{algo:weakly_polynomial_conv}. In each round $q$, we map each value $f(S)$ to $\big\lceil \frac{2f(S)}{\eps q} \big\rceil$, if it does not exceed $\lceil \frac{4}{\eps} \rceil$, where $S \subseteq [n]$; otherwise, it is set to $\infty$. Then, we run the exact algorithm with the maximum value as $\lceil \frac{4}{\eps} \rceil$. Once this convolution has been computed, we obtain the final result, by multiplying the elements by $\frac{\eps q}{2}$ and updating the minimum in each round.

Let $h(S) = f(T) + g(S \setminus T)$ be the value of the actual min-sum subset convolution. By construction, there is a round $q$ in which $q \leq h(S) < 2q$. At this point, note that the values $\lceil \frac{2f(T)}{\eps q} \rceil$ and $\lceil \frac{2 g(S \setminus T)}{\eps q} \rceil$ are bounded above by $\lceil \frac{4}{\eps} \rceil$, thus the subset $S$, regarded as an index in $\tilde h$, will be updated. Now consider the previous Lemma~\ref{lemma:weakly_helper}: Due to $q \leq h(S)$, the assumption is satisfied, and thus it holds that the sum lies between $h(S)$ and $(1+\eps)h(S)$. Since in the following rounds it still holds that $q \leq h(S)$, any updates at the subset $S$ will remain valid.

Let us analyze the running time of the algorithm. There are $O(\log M)$ iterations and the exact min-sum subset convolution algorithm runs in $\widetilde O(2^n \lceil \frac{4}{\eps} \rceil)$-time. This results in a total running time of $\widetilde O(2^n \log M / \eps)$.
\end{proof}

\section{Further Applications}\label{appendix:further_apps}

Our extensive list of applications underscores the importance of subset convolution in semi-rings, and the need for an approximate counterpart to alleviate the high running times of the exact solutions.

\subsection{Prize-Collecting Steiner Tree}\label{appendix:prize_collecting}

The Steiner tree problem is one of the well-known \NP-hard problems~\cite{dreyfus1971steiner}. In its generalization, the \emph{prize-collecting} Steiner tree problem (PCSTP), one penalizes the nodes not included in the chosen tree. Given an undirected graph $G = (V, E)$ on $n$ vertices, edge weights $c : E \to \mathbb{Q}_{> 0}$, and node weights (or \emph{prizes}) $p : V \to \mathbb{Q}_{> 0}$, we search for a  tree $S = (V(S), E(S)) \subseteq G$ such that
\begin{equation*}
    C(S) := \displaystyle\sum_{e\in E(S)}c(e) + \displaystyle\sum_{v\in V \setminus V(S)} p(v)
\end{equation*}
is minimized. We note that each Steiner tree instance can be transformed into a PCSTP instance by setting sufficiently large node weights for the terminals.\footnote{Recall that for the minimum Steiner tree problem, one enforces that the tree contains a given set of terminals.}

PCSTP was introduced in the 11th DIMACS Challenge (2014), and since then many solvers have been introduced in the literature pushing the boundaries of what was possible during the challenge~\cite{pcstp_ref_20, pcstp_ref_22, pcstp_ref_37, pcstp_ref_47}. The first FPT-algorithm has been recently proposed by Rehfeldt and Koch~\cite{rehfeldt2022exact} and runs in time $\widetilde O(3^{s^+} n)$, where $s^+$ is the number of proper potential terminals (we will define it later). The problem has witnessed a rich line of research on approximation algorithms~\cite{pcstp_approx_ref_26, pcstp_approx_ref_19, pcstp_approx_ref_31, pcstp_approx_final}, culminating with an $(2-\eps)$-approximation due to Archer et al.~\cite{pcstp_approx_final}. Note that these approximation algorithms run in polynomial time. 

\sparagraph{Dynamic Program.} The most recent exact algorithm is due to Rehfeldt and Koch~\cite{rehfeldt2022exact} and runs in time $\widetilde O(3^{s^+})$-time, where $s^+$ is the number of proper potential terminals (see below). Their algorithm resembles the dynamic program for the ``vanilla'' STP, due to Dreyfus and Wagner~\cite{dreyfus1971steiner}. Note that STP itself has a faster algorithm in $O^*(c^{|T|})$-time, for any $c > 1$~\cite{fuchs2007dynamic}. As regards PCSTP, we are not aware of any adaptation of that algorithm to the prize-collecting setting.
To define $s^+$, let $T_p \vcentcolon= \{v \in V \mid p(v) > 0\}$ be the set of \emph{potential terminals}. A terminal $t \in T_p$ is \emph{proper} if
\[ p(t) > \displaystyle\min_{e \in \delta(\{t\})} c(e),\]
where $\delta(U) := \big\{\{u, v\} \in E \mid u \in U, v \in V \setminus U\big\}$ for $U \subseteq V$.

\sparagraph{Layer-wise Optimization}. Rehfeldt and Koch~\cite{rehfeldt2022exact} propose to first solve the rooted variant, RPCSTP, which incorporates the additional constraint that a non-empty set $T_f \subseteq V$, called that of \emph{fixed} terminals, must be part of any feasible solution. For this variant, they provide an $\widetilde O(3^{|T_f|}n)$-time algorithm. Then, by a transformation from PCSTP to an equivalent RPCSTP that has no proper potential terminals and satisfies $|T_f| = s^+ + 1$, one can solve the original problem in the same running time, namely $\widetilde O(3^{s+} n)$-time, using a recursion similar to that of Dreyfus and Wagner for the minimum Steiner tree: Choose an arbitrary $t_0 \in T_f$, as defined above, and let $T_f^- \vcentcolon=T_f \setminus \{t_0\}$. Then, for $i = 2, \ldots, |T_f^-|$ define the functions $\phi$ and $\gamma$ recursively as follows~\cite{rehfeldt2022exact}:\footnote{The authors denote the functions by $f$ and $g$. We slightly rename them to avoid any confusion with the functions used in the actual convolution.}
\begin{flalign*}
\phi(T_i, w) &= \displaystyle\min_{T\subsetneq T_i \mid T \neq \varnothing}\:\big(\gamma(T, w) + \gamma(T_i \setminus T, w)\big) - \displaystyle\sum_{u\in V\setminus\{w\}} p(u),\\
\gamma(T_i, v) &= \displaystyle\min_{u \in V}\:\big(\phi(T_i, u) + d'_{pc}(v, u)\big),
\end{flalign*}
with $T_i \subseteq T_f^-$ and $|T_i| = i$. The function $d'_{pc}(v, u)$ is the prize-constraint distance between $v$ and $u$. Since we are only focusing on the \emph{shape} of the recursion, and the distance function is not involved in the actual convolution, we refer the reader to Ref.~\cite{rehfeldt2022exact} for more details.

\sparagraph{$(1 + \eps)$-Approximation.} Obtaining the $(1 + \eps)$-approximation scheme is now a simple exercise. Namely, note that $\phi$ is indeed a subset convolution performed \emph{layer-wise} based on the cardinality of $T_i$. Thus, there are $s^{+}n$ convolutions performed in total to evaluate $\gamma(T_f^-, t_0)$, which stores the value of the optimal solution. Hence, we can use the same approach as for the previous problem and obtain:
\begin{theorem}
If $(1 + \eps)$-approximate min-sum subset convolution runs in time $T(n, \eps)$, then Prize-Collecting Steiner Tree can be $(1 + \eps)$-approximated in $\widetilde O(T(s^{+}, \eps / s^{+}n))$-time.
\end{theorem}

\begin{corollary} $(1 + \eps)$-Approximate prize-collecting Steiner tree can be solved in $\widetilde O(2^\frac{3s^+}{2} / \sqeps)$-time.    
\end{corollary}

\subsection{Computational Biology}\label{appendix:biology_apps}

We outline another application, this time in computational biology, and briefly touch on another proposed by Bj\"orklund et al.~\cite{fsc} in their work on fast subset convolution. We choose the former because of its intriguing inapproximability result (see Thm.~\ref{thm:thm_restated_colorful_subtrees}) and to show an application of \emph{max-sum} subset convolution for which we devise a natural $(1 - \eps)$-approximation scheme given our previous results.

\sparagraph{Mass Spectrometry.} B\"ocker and Rasche~\cite{boecker_rasche} consider the maximum colorful subtree problem, defined in the following, as an application for the \emph{de novo} interpretation of metabolite tandem mass spectrometry data.

\begin{definition}[\textsc{Maximum Colorful Subtree}]Given a vertex-colored DAG $G = (V, E)$ with colors $C$ and weights $w : E \to \mathbb{R}$, find the induced colorful subtree $T = (V_T, E_T)$ of $G$ of maximum weight $w(T) := \displaystyle\sum_{e \in e_T}w(e)$.
\end{definition}

We chose this particular problem due to its intriguing hardness result, obtained by Rauf et al.~\cite[Thm.~1]{colorful_subtrees}: Consider the relaxation where we drop the color constraints, resulting in the \textsc{Maximum Subtree} problem. Then,

\begin{theorem}[{\cite[Thm.~1]{colorful_subtrees}}]
\label{thm:thm_restated_colorful_subtrees}
There is no $O(|V|^{1 - \delta})$ (polynomial-time) approximation algorithm for the \prb{Maximum Subtree} problem for any $\delta > 0$, unless \textup{P} = \textup{NP}.    
\end{theorem}

Note that the inapproximability result naturally holds for the \textsc{Maximum Colorful Subtree} as well, since the latter is a generalization of the former. In this light, Rauf et al.~\cite{colorful_subtrees} designed an $\widetilde{O}(3^k k |E|)$-time dynamic program, with $k$ the number of colors, as follows. Let $W(v, S)$ be the maximal score of a colorful tree with root $v$ and color set $S \subseteq C$. The table $W$ can be computed by the following recursion~\cite{colorful_subtrees}:
\begin{subnumcases}{W(v, S)= \max}
   \displaystyle\max_{u:c(u) \in S \setminus \{c(v)\}, vu \in E}\big(W(u, S \setminus \{c(v)\}) + w(v, u)\big),\label{eq:color_subtree_1st}
   \\
   \displaystyle\max_{\substack{(S_1, S_2):S_1\cap S_2 = \{c(v)\}\\S_1 \cup S_2 = S}}\big(W(v, S_1) + W(v, S_2)\big).\label{eq:color_subtree_2nd}
\end{subnumcases}
They remark that the running time can be indeed improved to $O(2^n \poly(|V|, k))$, assuming the problem instance enjoys a bounded input; that is, using the standard embedding technique to embed the $(\max, +)$-semi-ring into the $(+, \times)$-ring proposed by Bj\"orklund et al.~\cite{fsc}.

Note that Eq.~\eqref{eq:color_subtree_2nd} indeed is a max-sum subset convolution. To see why this is the case, observe that the expression $W(v, S)$ can be computed via \emph{max-sum} subset convolution by a layer-wise optimization, i.e., evaluate the convolution for all sets $S$ of a given cardinality (this is indeed a standard approach undertaken by Bj\"orklund et al.~\cite{fsc} to optimize recursive dynamic programs). To avoid the constraint $S_1 \cap S_2 = \{c(v)\}$, we can introduce $S_1' = S_1 \setminus \{c(v)\}$ and $S_2' = S_2 \setminus \{c(v)\}$, with $S_1'\:\dot{\cup}\:S_2' = S - \{c(v)\}$, and compute the subset convolution for $S \setminus \{c(v)\}$ instead. Certainly, we also need to define the function $f$ used in the convolution so that the weights are the correct ones. Namely, we set $f(X) := W(v, X \cup \{c(v)\})$, for any set $X \subseteq S \setminus \{c(v)\}$.

Our strategy developed for the above problems can be applied to this problem as well, with the only impediment that the convolution is in the $(\max, +)$-semi-ring. Extending our results to this convolution, however, is a simple exercise since maximizing any function $\phi$ is the same as minimizing $-\phi$. There is a small caveat here: the original framework requires that the values are positive. To ensure this, we simply need to shift $-\phi$ with the largest value (with a certain epsilon to ensure strict positivity). We hence obtain:

\MaxSumSubsetConvolution*
Note that the weakly-polynomial approximation algorithm can also be adapted to work for the max-sum subset convolution (a similar case holds in the sequence setting~\cite{partition_karol}).
Since the convolution is called at most $O(k|V|)$ times in the evaluation of the dynamic program, we obtain an $(1 - \eps)$-approximation for our problem:

\begin{theorem}
If $(1 - \eps)$-Approximate max-sum subset convolution can be solved in time $T(n, \eps)$, then maximum colorful subtree can be $(1 - \eps)$-approximated in time $O(T(k, \eps / k|V|))$.
\end{theorem}

\begin{corollary}
$(1-\eps)$-Approximate maximum colorful subtree can be solved in $\widetilde O(2^\frac{3k}{2} / \sqeps)$-time.
\end{corollary}

\sparagraph{Protein Interaction Networks.} A similar approach can be considered for the problem of detecting signaling pathways in protein interaction networks~\cite{protein_network}, considered for application by Bj\"orklund et al.~\cite{fsc}. We briefly mention that the problem can also be $(1 + \eps)$-approximated in time $\widetilde O(2^\frac{3k}{2} / \sqeps)$. We skip the details and refer the interested reader to Ref.~\cite{fsc}.

\subsection{Supplementary Applications}

There are many other applications that use the min/max-sum subset convolution and inherently resort to its \emph{exact} $\widetilde O(2^n M)$-time algorithm, such as in Bayesian network learning~\cite{bayesian_network_learning} and the convex recoloring problem~\cite{convex_recoloring}. Our proposal of \emph{approximate} min/max-sum subset convolution thus enables approximate counterparts of these particular problems.

\section{Missing Proofs}

In this section, we provide the missing proofs from the main paper.

\subsection{Simple Strongly-Polynomial Approximation Algorithm}\label{appendix:simple_strongly}

\SimpleStrongly*
\begin{proof}
To apply the Sum-to-Max-Covering lemma, we view the set functions $f$ and $g$, defined on the subset lattice of order $n$, as vectors in $\mathbb{R}^{2^n}_{+}$. Indeed, the coordinate of a set $S \subseteq [n]$ is its implicit binary representation. This yields the vectors $f^{(1)}, \ldots, f^{(s)}, g^{(1)}, \ldots, g^{(s)} \in \mathbb{R}^{2^n}_{+}$, with $s = O(\reveps \log \reveps + n \log \reveps)$. Then, for each $\ell \in [s]$, we perform the \emph{exact} min-max subset convolution of $f^{(\ell)}$ and $g^{(\ell)}$. At this point, we do a remark that has also been done by Bringmann et al.~\cite{approx_min_plus}: We first need to replace the entries of $f^{(\ell)}$ and $g^{(\ell)}$ by their ranks, respectively, since min-max subset convolution requires, as in the case of Kosaraju's algorithm, standard bit representation of integers. Once $h^{(\ell)} = f^{(\ell)} \ovee g^{(\ell)}$ has been computed, we can take the entry-wise minimum over $\ell \in [s]$, as outlined in Alg.~\ref{alg:approx-minconv-simple}; we use $\texttt{idx}(S)$ to denote the binary representation of the set $S$, using it to obtain the corresponding coordinate in an $2^n$-dimensional vector.

The total running time accumulates to $\widetilde O(s 2^\frac{3n}{2}) = \widetilde O(2^\frac{3n}{2} / \eps)$, since the construction of the auxiliary vectors takes linear time in the output size.\footnote{Precisely, the construction of the auxiliary vectors takes time $O(\frac{2^n}{\eps} \log \reveps + 2^n n \log \reveps)$.}
\end{proof}

\subsection{Equivalence of Exact Min-Max and Approximate Min-Sum}\label{appendix:equivalence}

\Equivalence*
\begin{proof}
To prove the theorem, we can employ similar techniques to the those introduced by Bringmann et al.~\cite[Thm.~1.5]{approx_min_plus}. Namely, for the second direction, we use our algorithm from Lemma~\ref{lemma:simple_approx_min_sum_subset_conv}, which approximates the min-sum subset convolution via the exact min-max subset convolution. If $T(n)$ is the running time of the exact min-max subset convolution, then the approximation algorithm runs in time $\widetilde O(T(n) / \eps)$. For the first direction, let $h = f \ovee g$ and $t \vcentcolon= \ceil{4(1 + \eps)^2}$. Consider the functions $f'$ and $g'$ defined as $f'(S) \vcentcolon=t^{f(S)}$ and $g'(S) \vcentcolon= t^{g(S)}$ for $S \subseteq [n]$, respectively.\footnote{Since $f(S)$ and $g(S)$ are integers in the standard bit representation, we can compute $t^{f(S)}$ and $t^{g(S)}$ by simply writing $f(S)$ and $g(S)$, respectively, into the exponent.} Let $h'$ be the $(1 + \eps)$-approximate min-sum subset convolution of $f'$ and $g'$.
\begin{restatable}[{cf.~\cite{approx_min_plus}}]{claim}{ApproxClaim}It holds $t^{h(S)} \leq h'(S) \leq t^{h(S) + 1/2}$ for all $S \subseteq [n]$.
\label{claim:approx}
\end{restatable}
Using the above claim (we prove it below), we can compute $h(S)$ as $\floor{\log_t h'(S)}$.\footnote{This can be achieved by simply reading the exponent of $h'(S)$.} If $(1 + \eps)$-approximate min-sum subset convolution runs in time $T(n)$, we obtain an algorithm for exact min-max subset convolution that runs in time $O(2^n + T(n)) = \widetilde O(T(n))$.
\end{proof}

Let us prove Claim~\ref{claim:approx}:

\ApproxClaim*
\begin{proof}
Recall that $h = f \ovee g$ and $h'$ is the $(1 + \eps)$-approximate min-sum subset convolution of $f', g'$, defined as $f'(S) \vcentcolon= t^{f(S)}, g'(S) \vcentcolon= t^{g(S)}$, where $t \vcentcolon= \ceil{4(1 + \eps)^2}$ and $\eps > 0$. We employ the same proof strategy as in \cite[Claim 4.1]{approx_min_plus}, originally designed for the APSP problem.

Fix $S \subseteq [n]$ and observe that
\[
    \displaystyle\min_{T\subseteq S}\left(f'(T) + g'(S \setminus T)\right) \leq h'(S) \leq (1+\eps)\displaystyle\min_{T \subseteq S}\left(f'(T) + g'(S \setminus T)\right)
\]
holds by definition. Furthermore, there exists an $T^*$ such that $h(S) = \max\{f(T^*), g(S \setminus T^*)\}$. Thus, we have
\begin{flalign*}
    h'(S) &\leq (1 + \eps)(f'(T^*) + g'(S\setminus T^*)\\
    &= (1 + \eps)(t^{f(T^*)} + t^{g(S\setminus T^*)})\\
    &\leq 2(1 + \eps)t^{\max\{f(T^*), g(S \setminus T^*)\}}\\
    &= 2(1+\eps)t^{h(S)}\\
    &\leq t^{h(S) + 1/2}, \text{by $t \geq 4(1+\eps)^2$},
\end{flalign*}
proving the upper-bound. To show the lower-bound, observe that there is an $T \subset S$ with $h'(S) \geq f'(T) + g'(S \setminus T)$. Hence,
\begin{flalign*}
    h'(S) &\geq f'(T) + g'(S \setminus T)\\
    &= t^{f(T)} + t^{g(S \setminus T)}\\
    &\geq t^{\max\{f(T), g(S \setminus T)\}}\\
    &\geq t^{h(S)}.
\end{flalign*}
\end{proof}

\subsection{Distant Summands Lemma}

We provide the full proof of Lemma~\ref{lemma:distant_summands}, which takes care of the \emph{distant} summands case.

\begin{algorithm}
    \caption{$\textsc{DistantConv}(f,g,\eps)$ [Lemma~\ref{lemma:distant_summands}]}
	\label{algo:distant_conv}
\begin{algorithmic}[1]
    \State $\{f^{(1)},\ldots,f^{(s)},g^{(1)},\ldots,g^{(s)}\} = \textsc{DistantCovering}(f, g, \frac{\eps}{4})$ (see~\cite[Cor.~5.10]{approx_min_plus})
    \State $h^{(\ell)} \gets \textsc{MinMaxSubsetConv}(f^{(\ell)}, g^{(\ell)})$ for any $\ell \in [s]$
    \State $\widetilde h(S) \gets \frac{1}{1-\eps/2} \cdot \displaystyle\min_{\ell \in [s]}\:\{ h^{(\ell)}[\texttt{idx}(S)] \}$ for all $S \subseteq [n]$
    \State \Return $\widetilde h$
\end{algorithmic}
\end{algorithm}
\DistantSummandsLemma*
\begin{proof}
    To this end, let us consider Alg.~\ref{algo:distant_conv}. We view the set functions $f, g$ as vectors in $\mathbb{R}^{2^n}_{+}$. Hence, we can use Distant Covering from the framework~\cite[Cor.~5.10]{approx_min_plus} on $f, g$ with $\eps' \vcentcolon= \frac{\eps}{4}$ and obtain $f^{(1)}, \ldots, f^{(s)}, g^{(1)}, \ldots, g^{(s)}$ with $s = O\left(\polylog(\frac{2^n}{\eps})\right)$.\footnote{For completeness, we provide \cite[Cor.~5.10]{approx_min_plus} as Lemma~\ref{lemma:distant_covering_lemma} in Appendix~\ref{appendix:distant_covering_lemma}.} For each $l \in [s]$, we compute $h^{(l)} = f^{(l)} \ovee g^{(l)}$ and scale the entry-wise minimum by $\frac{1}{1 - 2\eps'}$. The scaling factor removes the extra factor $1 - 2\eps'$ from the RHS of property \texttt{(i)} of \cite[Cor.~5.10]{approx_min_plus}, yielding $\widetilde h(S) \geq h(S)$ for any $S \subseteq [n]$. Then, using property~\texttt{(ii)} of the same corollary, we have that for any indices $i, j$ with $\frac{{f}[i]}{{g}[j]} \not\in [\eps', \frac{1}{\eps'}] = [\frac{\eps}{4}, \frac{4}{\eps}]$ there is an $\ell$ such that $h^{(l)}[i + j] \leq {f}[i] + {g}[j]$. Looking at this from the perspective of subset convolution, we obtain that for any $T \subseteq S$ with $\frac{f(T)}{g(S \setminus T)} \not\in [\frac{\eps}{4}, \frac{4}{\eps}]$ there is such an $\ell$ with $h^{(l)}(S) = f(T) + g(S \setminus T)$. Hence, minimizing over all $\ell$ and scaling by $\frac{1}{1 - \eps / 2} < 1 + \eps$ gives the second property of our lemma.
    
    The running time is dominated by the application of the min-max subset convolution $s$ times, leading to a total running time of $\widetilde O(2^\frac{3n}{2} s) = O\left(2^\frac{3n}{2} \polylog(\frac{2^n}{\eps})\right)$.
\end{proof}

\subsection{Close Summands Lemma}\label{appendix:close_summands_lemma}

The following lemma of Bringmann et al.~\cite{approx_min_plus} is helpful in the proof for \emph{close} summands.
\begin{restatable}[{\cite[Lemma 8.6]{approx_min_plus} and \cite[Lemma B.2]{partition_karol}}]{lemma}{ApproxLemma}
    \label{lemma:rounding}
For any $x, y, q, \eps \in \mathbb{R}_{+}$ with $x + y \geq q/2$ and $\eps \in (0, 1)$ we have:
\[
    x+y \leq \left(\ceil{\frac{4x}{q\eps}} + \ceil{\frac{4y}{q\eps}}\right) \frac{q\eps}{4} \leq (1+\eps)(x+y).
\]
\end{restatable}
\begin{proof}
The lower-bound follows immediately. To show the upper-bound, note that
\[
\left(\ceil{\frac{4x}{q\eps}} + \ceil{\frac{4y}{q\eps}}\right)\frac{q\eps}{4} \leq x + y + 2\frac{q\eps}{4} = x + y + \eps\frac{q}{2} \leq (1+\eps)(x+y).
\]
\end{proof}
We now prove Lemma~\ref{lemma:close_summands}, which takes care of the \emph{close} summands case.

\begin{algorithm}
    \caption{$\textsc{CloseConv}(f,g,\eps)$ [Lemma~\ref{lemma:close_summands}]}
	\label{algo:close_conv}
\begin{algorithmic}[1]
    \Procedure{Scale}{$f, q, \eps$}
    	\State \Return $S \mapsto
            \begin{cases}\big\lceil \frac{4}{\eps q}\cdot f(S) \big\rceil, & \text{if} \; \frac{\eps q}{16} \le f(S) \le q,\\
    		\infty, & \text{otherwise.}\\
    	\end{cases}$
    \EndProcedure
    \State
    \State Set $\widetilde{h}(S) \gets \infty$ for all $S \subseteq [n]$
    \For {$q \in \{1,2,4,\ldots, 2^{\ceil{\log{2M}}}\}$}
      \State $f_q, g_q \gets \textsc{Scale}(f,q,\eps), \textsc{Scale}(g,q,\eps)$
      \State $h_q \gets \textsc{ExactMinSumSubsetConv}(f_q,g_q)$ 
      \State $\widetilde{h}(S) \gets \min \{\widetilde{h}(S), h_q(S) \cdot \frac{\eps q}{4} \}$ for all $S \subseteq [n]$
    \EndFor
    \State \Return $\widetilde{h}$
\end{algorithmic}
\end{algorithm}

\CloseSummandsLemma*
\begin{proof}
To this end, let us consider Alg.~\ref{algo:close_conv}. Its correctness can be shown in a similar way as in Ref.~\cite{approx_min_plus} (or equivalently Ref.~\cite{partition_karol}), namely: Property \texttt{(i)} follows directly from the lower-bound provided by Lemma~\ref{lemma:rounding}. To show property \texttt{(ii)}, we make the following observation: For any $T \subseteq S$ with $h(S) = f(T) + g(S \setminus T)$ there exists $q$ with $\frac{q}{2} \leq f(T) + g(S \setminus T) \leq q$. This implies $f(T), g(S \setminus T) \leq q$ and $\max\{f(T), g(S \setminus T)\} \geq \frac{q}{4}$. Now, with the additional assumption $\frac{f(T)}{g(S \setminus T)} \in [\frac{\eps}{4}, \frac{4}{\eps}]$, we have
\[
    \min\{f(S), g(S \setminus T)\} \geq \frac{\eps}{4}\max\{f(T), g(S \setminus T)\} \geq \frac{\eps q}{16},
\]
and hence, in the $q$th round, neither $f_q(T)$ nor $g_q(S \setminus T)$ are set to $\infty$. Finally, using the upper-bound of Lemma~\ref{lemma:rounding}, we obtain $\widetilde h(S) \leq (1+\eps)\left(f(T) + g(S \setminus T)\right) = (1 + \eps)h(S)$.

Let us analyze its running time. To this end, let $\alpha_q$ be the number of entries of $f_q$ that are not set to $\infty$, i.e., $\alpha_q = \{S \mid f_q(S) < \infty\}$; define $\beta_q$ similarly. We show that the exact min-sum subset convolution of $f_q$ and $g_q$ can be performed in time $O\left(\min\{\alpha_q\beta_q, \frac{2^n}{\eps}\}\polylog(\frac{2^n}{\eps})\right)$ as follows.

First, we analyze lines 7-8. To maintain the intermediate functions $f_q$ and $g_q$, we use an event queue to skip the rounds with $\alpha_q = 0$ or $\beta_q = 0$, since in such cases the min-sum subset convolution is filled with $\infty$ and would have been skipped in line 10 anyway. Additionally, we can perform the transitions $f_q \rightarrow f_{q + 1}$, $g_q \rightarrow g_{q + 1}$, i.e., between two consecutive iterations, in time proportional to the number of non-$\infty$ entries, i.e., $O(\alpha_q + \beta_q)$, which is l.e.q.~$O(\min\{\alpha_q\beta_q, \frac{2^n}{\eps}\})$.

There are now two ways to compute the exact convolution of $f_q$ and $g_q$.~First, naively performing the convolution between the non-$\infty$ entries takes time $O(\alpha_q \beta_q)$: Consider \emph{disjoint} $T_1, T_2 \subseteq [n]$ such that $f_q(T_1) \neq \infty$ and $g_q(T_2) \neq \infty$.~Then, we can update $h_q(T_1 \cup T_2)$ with the value of $f_q(T_1) + g_q(T_2)$, if smaller. The second way exploits the fact that the values in $f_q$ and $g_q$ are \emph{bounded}. Indeed, consider the procedure \textsc{Scale} (within Alg.~\ref{algo:close_conv}) which maps $f(U)$ to $\ceil{\frac{4}{\eps q} f(U)}$, if $\frac{\eps q}{16} \leq f(U) \leq q$. This imposes that $f_q$ only takes values bounded by $\ceil{\frac{4}{\eps}}$; analogously for $g_q$. We can hence use the $\widetilde O(2^n M)$-time algorithm for bounded input, as described by Bj\"orklund et al.~\cite{fsc} (and mentioned several times in our motivation), by setting $M = \ceil{\frac{4}{\eps}}$.\footnote{Indeed, this running time corresponds to that of the FFT-based solution in the sequence setting~\cite[Lemma 5.7.2]{karol_phd}, namely $\widetilde O(n W)$.} The running time can thus be bounded by
\[
\displaystyle\sum_{q} \min\{\alpha_q\beta_q, 2^n / \eps\}\polylog(2^n / \eps),
\]
by using the better of the two ways. As in Ref.~\cite{approx_min_plus}, we introduce a threshold $\lambda$ and split the previous sum into two cases
w.r.t.~$\lambda$, i.e., $\beta_q \leq \lambda$ and $\beta_q > \lambda$.
At this point, observe that the second case can occur at most $O(\frac{2^n}{\lambda}\log \reveps)$ times. To see why this is the case, consider an entry $f(U)$ for $U \subseteq [n]$. Since $q$ grows geometrically, there are $O(\log \reveps)$ rounds in which $f(U)$ is not set to $\infty$. Hence, we obtain $\sum_{q} \alpha_q = O(2^n \log \reveps)$; similarly, $\sum_{q} \beta_q = O(2^n \log \reveps)$. With this observation, we can bound the running time by
\begin{flalign*}
    \left(\displaystyle\sum_{\beta_q \leq \lambda} \alpha_q\beta_q + \displaystyle\sum_{\beta_q > \lambda} \frac{2^n}{\eps}\right) \polylog(2^n / \eps)
    &\leq \left(\lambda \cdot 2^n \log \reveps + \frac{2^n}{\eps}\frac{2^n}{\lambda}\log \reveps \right) \polylog(2^n / \eps)\\
    &\leq \left(\lambda \cdot 2^n + \frac{(2^n)^2}{\eps \lambda}\right) \polylog(2^n / \eps).
\end{flalign*}
The expression is minimized for $\lambda = \left(2^n / \eps \right)^{1/2}$, yielding the total running time of $\widetilde O(2^\frac{3n}{2} / \sqeps)$.
\end{proof}

\subsection{Strongly-Polynomial Approximation Algorithm}\label{appendix:algo_extra}

We outline the algorithm underlying Thm.~\ref{thm:strongly_approx_min_sum_subset_convolution}.

\begin{algorithm}
    \caption{$\textsc{ApxMinSumSubsetConv}(f,g,\eps)$ [Thm.~\ref{thm:strongly_approx_min_sum_subset_convolution}] -- Running time: $\widetilde O(2^\frac{3n}{2} / \sqeps)$}
    \label{algo:apx_min_sum_subset_conv_final}
\begin{algorithmic}[1]
    \State $\widetilde{h}_1 \gets \textsc{DistantConv}(f,g,\eps)$
    \State $\widetilde{h}_2 \gets \textsc{CloseConv}(f,g,\eps)$
    \State $\widetilde h(S) \gets \min  \{ \widetilde{h}_1(S), \widetilde{h}_2(S) \}$ for all $S \subseteq [n]$
    \State \textbf{Return} $\widetilde{h}$
\end{algorithmic}
\end{algorithm}

\section{Auxiliary Results}

In this section, we state for completeness several results from the framework of Bringmann et al.~\cite{approx_min_plus} that we used throughout our paper, along with their corresponding references.

\subsection{Distant Covering Lemma}\label{appendix:distant_covering_lemma}

\begin{lemma}[{Distant Covering~\cite[Cor.~5.10]{approx_min_plus}}]
\label{lemma:distant_covering_lemma}

Given vectors $A, B \in \mathbb{R}^{d}_{+}$ and a parameter $\eps > 0$, there exist vectors $A^{(1)}, \ldots, A^{(s)}, B^{(1)}, \ldots, B^{(s)} \in \mathbb{R}^d_+$ with $s = O(\log d \log \reveps)$ such that
\begin{enumerate}[(i)]
    \item for all $i, j \in [d]$ and all $l \in [s]$: $\max\{A^{(l)}[i], B^{(l)}[j]\} \geq (1-2\eps)(A[i] + B[j])$, and
    \item for all $i, j \in [d]$ if $\frac{A[i]}{B[j]} \not\in [\eps, 1/\eps]$ then $\exists l \in [s]$: $\max\{A^{(l)}[i], B^{(l)}[j]\} \leq A[i] + B[j]$.
\end{enumerate}
The auxiliary vectors $A^{(1)}, \ldots, A^{(s)}, B^{(1)}, \ldots, B^{(s)}$ can be computed in time $O(d \log d \log \reveps)$.
\end{lemma}

\section{Input Format}\label{appendix:input_format}

For simplicity, we will denote by $n$ the input size corresponding to the problem. Note that the natural input size for subset convolutions on a subset lattice of order $k$ is indeed $2^k$.

We follow the setup introduced by Bringmann et al.~\cite{approx_min_plus}: Input numbers in all \emph{approximate} problems are represented in floating-point, while those in \emph{exact} problems are represented in the standard bit representation. The choice of this particular input format is motivated by the use of floating-point operations in the equivalences of approximate min-sum subset convolution and exact min-max subset convolution (see Thm.~\ref{thm:equivalence}).

Namely, in the reduction from exact min-max to approximate min-sum, we need to \emph{exponentiate} the given numbers. However, under the standard bit representation, this requires shifting from $O(\log n)$-bit integers to $\poly(n)$-bit integers. Understandably, this is not efficient enough. The main observation of Bringmann et al.~\cite{approx_min_plus} is that if $m$ is an $O(\log n)$-bit integer in bit-representation, then one can store $2^m$ as a floating-point number by storing $m$ as the exponent. This implies that the floating-point number has an $O(\log n)$-bit exponent (and an $O(1)$-bit mantissa). For the other direction, i.e., from approximate min-sum in floating-point representation to exact min-max in bit representation, we use the fact that for exact min-max we can replace the input numbers by their \emph{ranks}, yielding $O(\log n)$-bit integers in bit representation~\cite{approx_min_plus}.

\sparagraph{Bit Representation.} The standard bit representation is used for the exact min-max subset convolution. The key observation is that we can replace the actual input numbers by their \emph{ranks}, the position in the sorted order of \emph{all} input numbers (this takes near-linear time and, in our context, does not modify the time-complexity of the algorithms). This construction ensures that the numbers are integers in $\{1, \ldots, \polylog(n)\}$ and, hence, take $O(\log n)$ bits. 

\sparagraph{Floating-Point Representation.} As regards the $(1 + \eps)$-approximate problems, we assume that the input is given in floating-point format, following the setup proposed Bringmann et al.~\cite{approx_min_plus}. In particular, if the input numbers are in the range $[1, M]$, it suffices to store for each input number $w$ its rounded algorithm $e = \floor{\log_2 w}$, requiring $O(\log \log M)$ bits, and an approximation of the mantissa, i.e., $w / 2^e$, requiring only $O(\log \frac{1}{\eps})$ bits, yielding a total of $O(\log \log M + \log \frac{1}{\eps})$ bits. We refer the reader to Ref.~\cite{approx_min_plus} for more details.

\section{Frank-Tardos' Framework and Lossy Kernels}\label{appendix:lossy_kernels}

Frank-Tardos' framework~\cite{frank_tardos} is known for its ability to reduce the input values in weighted \NP-hard problems so that the algorithm solving the original instance becomes strongly polynomial instead of weakly polynomial~\cite{poly_kernel_etscheid}. In fact, on closer inspection, e.g. using the construction in ~\cite[Example~1.1]{poly_kernel_bentert}, we can also obtain a polynomial kernel for the minimum-cost $k$-coloring problem (we inspected this for $k = 2$).

However, the instance constructed in this way is only guaranteed to preserve the \emph{comparison outputs} between the (partial) solution values. There is no guarantee on the approximation of the values obtained in this way after running the approximation algorithm. Therefore, we cannot \emph{directly} use Frank-Tardos' as is before running our weakly polynomial approximation algorithm, Thm.~\ref{thm:weakly_approx_min_sum_subset_convolution}. We believe that using $\alpha$-approximate kernels~\cite{lossy_kernels} may be the way forward. The first step is to understand whether our applications accept such kernels, and whether they are of polynomial size, e.g., Dvor\'ak et al.~\cite{pcst_lossy} tackled many variants of the Steiner tree problem, but the prize-collecting, though mentioned, remained unexplored. This is orthogonal to our proposal and left as an intriguing future work.

\end{document}